\documentclass[12pt]{article}%

\makeatletter
\def\input@path{{styles/}}
\makeatother

\def\UseBibLatex{1}

\makeatletter
\def\input@path{{styles/}}
\makeatother

\providecommand{\BibLatexMode}[1]{}
\providecommand{\BibTexMode}[1]{}

\renewcommand{\BibLatexMode}[1]{#1}
\renewcommand{\BibTexMode}[1]{}

\ifx\UseBibLatex\undefined%
  \renewcommand{\BibLatexMode}[1]{}
  \renewcommand{\BibTexMode}[1]{#1}
\fi

\BibLatexMode{%
   \usepackage[bibencoding=auto,style=alphabetic,backend=biber]{biblatex}%
   \usepackage{sariel_biblatex}%
}

\usepackage[in]{fullpage}%
\usepackage{amsmath}%
\usepackage{amssymb}%
\usepackage[table]{xcolor}%

\usepackage[amsmath,thmmarks]{ntheorem}%

\usepackage{titlesec}%
\usepackage{xcolor}%
\usepackage{mleftright}%
\usepackage{xspace}%
\usepackage{graphicx}
\usepackage{hyperref}%
\usepackage[inline]{enumitem}
\usepackage{hyperref}%
\usepackage[ocgcolorlinks]{ocgx2}

\hypersetup{%
      unicode,
      breaklinks,%
      colorlinks=true,%
      urlcolor=[rgb]{0.25,0.0,0.0},%
      linkcolor=[rgb]{0.5,0.0,0.0},%
      citecolor=[rgb]{0,0.2,0.445},%
      filecolor=[rgb]{0,0,0.4},
      anchorcolor=[rgb]={0.0,0.1,0.2}%
}

\titlelabel{\thetitle. }%

\theoremseparator{.}%

\theoremstyle{plain}%
\newtheorem{theorem}{Theorem}[section]

\newtheorem{lemma}[theorem]{Lemma}

\theoremstyle{plain}%
\theoremheaderfont{\sf} \theorembodyfont{\upshape}%
\newtheorem*{remark:unnumbered}[theorem]{Remark}%

\newtheorem{defn}[theorem]{Definition}

\newtheorem{problem}[theorem]{Problem}

\theoremheaderfont{\em}%
\theorembodyfont{\upshape}%
\theoremstyle{nonumberplain}%
\theoremseparator{}%
\theoremsymbol{\myqedsymbol}%
\newtheorem{proof}{Proof:}%

\providecommand{\emphind}[1]{}%
\renewcommand{\emphind}[1]{\emph{#1}\index{#1}}

\definecolor{blue25emph}{rgb}{0, 0, 11}

\providecommand{\emphic}[2]{}
\renewcommand{\emphic}[2]{\textcolor{blue25emph}{%
      \textbf{\emph{#1}}}\index{#2}}

\providecommand{\emphi}[1]{}%
\renewcommand{\emphi}[1]{\emphic{#1}{#1}}

\definecolor{almostblack}{rgb}{0, 0, 0.3}

\providecommand{\emphw}[1]{}%
\renewcommand{\emphw}[1]{{\textcolor{almostblack}{\emph{#1}}}}%

\providecommand{\emphOnly}[1]{}%
\renewcommand{\emphOnly}[1]{\emph{\textcolor{blue25}{\textbf{#1}}}}

\newcommand{\myqedsymbol}{\rule{2mm}{2mm}}
\newcommand{\JonathanThanks}[1]{%
   \thanks{%
      Department of Computer Science; %
      University of Illinois; %
      201 N. Goodwin Avenue; %
      Urbana, IL, 61801, USA; %
      \href{mailto:spam@illinois.edu}{jed3@illinois.edu}.%
   #1%
   }%
}

\newcommand{\SarielThanks}[1]{%
   \thanks{%
      Department of Computer Science; %
      University of Illinois; %
      201 N. Goodwin Avenue; %
      Urbana, IL, 61801, USA; %
      \href{mailto:spam@illinois.edu}{sariel@illinois.edu}; %
      \url{http://sarielhp.org/}.%
   #1%
   }%
}

\newcommand{\HLink}[2]{\hyperref[#2]{#1~\ref*{#2}}}
\newcommand{\HLinkSuffix}[3]{\hyperref[#2]{#1\ref*{#2}{#3}}}

\newcommand{\figlab}[1]{\label{fig:#1}}
\newcommand{\figref}[1]{\HLink{Figure}{fig:#1}}

\newcommand{\thmlab}[1]{{\label{theo:#1}}}
\newcommand{\thmref}[1]{\HLink{Theorem}{theo:#1}}

\providecommand{\deflab}[1]{\label{def:#1}}

\newcommand{\defrefY}[2]{\hyperref[def:#2]{#1}}

\newcommand{\seclab}[1]{\label{sec:#1}}
\newcommand{\secref}[1]{\HLink{Section}{sec:#1}}

\newcommand{\lemlab}[1]{\label{lemma:#1}}
\newcommand{\lemref}[1]{\HLink{Lemma}{lemma:#1}}%

\providecommand{\eqlab}[1]{}%
\renewcommand{\eqlab}[1]{\label{equation:#1}}

\providecommand{\remove}[1]{}%
\newcommand{\Set}[2]{\left\{ #1 \;\middle\vert\; #2 \right\}}

\newcommand{\pth}[1]{\mleft(#1\mright)}%

\newcommand{\ProbC}{{\mathbb{P}}}
\newcommand{\ExC}{{\mathbb{E}}}

\newcommand{\Prob}[1]{\ProbC\mleft[ #1 \mright]}
\newcommand{\Ex}[1]{\ExC\mleft[ #1 \mright]}

\newcommand{\ceil}[1]{\mleft\lceil {#1} \mright\rceil}

\newcommand{\cardin}[1]{\left\lvert {#1} \right\rvert}%

\renewcommand{\th}{th\xspace}

\renewcommand{\Re}{\mathbb{R}}%

\newlist{compactenumA}{enumerate}{5}%
\setlist[compactenumA]{topsep=0pt,itemsep=-1ex,partopsep=1ex,parsep=1ex,%
   label=(\Alph*)}%

\newlist{compactenuma}{enumerate}{5}%
\setlist[compactenuma]{topsep=0pt,itemsep=-1ex,partopsep=1ex,parsep=1ex,%
   label=(\alph*)}%

\newlist{compactenumI}{enumerate}{5}%
\setlist[compactenumI]{topsep=0pt,itemsep=-1ex,partopsep=1ex,parsep=1ex,%
   label=(\Roman*)}%

\newlist{compactenumi}{enumerate}{5}%
\setlist[compactenumi]{topsep=0pt,itemsep=-1ex,partopsep=1ex,parsep=1ex,%
   label=(\roman*)}%

\newlist{compactitem}{itemize}{5}%
\setlist[compactitem]{topsep=0pt,itemsep=-1ex,partopsep=1ex,parsep=1ex,%
   label=\ensuremath{\bullet}}%

\numberwithin{figure}{section}%
\numberwithin{table}{section}%
\numberwithin{equation}{section}%

\DeclareFontFamily{U}{BOONDOX-calo}{\skewchar\font=45 }
\DeclareFontShape{U}{BOONDOX-calo}{m}{n}{
  <-> s*[1.05] BOONDOX-r-calo}{}
\DeclareFontShape{U}{BOONDOX-calo}{b}{n}{
  <-> s*[1.05] BOONDOX-b-calo}{}
\DeclareMathAlphabet{\mathcalb}{U}{BOONDOX-calo}{m}{n}
\SetMathAlphabet{\mathcalb}{bold}{U}{BOONDOX-calo}{b}{n}
\DeclareMathAlphabet{\mathbcalb}{U}{BOONDOX-calo}{b}{n}

\usepackage{stmaryrd}%

\newcommand{\IntRange}[1]{\left\llbracket #1 \right\rrbracket}
\newcommand{\IRX}[1]{\IntRange{#1}}
\newcommand{\IRY}[2]{\IntRange{#1:#2}}

\newcommand{\TwoCeil}[1]{\ceil{\!\!\ceil{ #1 }\!\!}}

\newcommand{\eps}{{\varepsilon}}%

\newcommand{\rect}{\Mh{\mathcalb{r}}}%
\newcommand{\qrect}{\Mh{\mathcalb{q}}}%

\newcommand{\lN}{\mathcalb{h}}

\newcommand{\cGrid}{\mathcalb{c}_1}

\renewcommand{\P}{\Mh{P}}%
\newcommand{\Q}{\Mh{Q}}%

\renewcommand{\S}{\Mh{S}}%
\newcommand{\etal}{\textit{et~al.}\xspace}

\newcommand{\Term}[1]{\textsf{#1}}

\newcommand{\LCA}{\Term{LCA}\xspace}%

\newcommand{\polylog}{\mathrm{polylog}}%

\newcommand{\Patrascu}{P{\u{a}}tra{\c{s}}cu\xspace}

\providecommand{\Mh}[1]{#1}%
\newcommand{\rankX}[1]{\mathrm{rank}\pth{#1}}
\newcommand{\rankY}[2]{\mathrm{rank}\pth{#1, #2}}
\providecommand{\TPDF}[2]{\texorpdfstring{#1}{#2}}
\newcommand{\areaX}[1]{\mathrm{area}\pth{#1}}
\newcommand{\RectSet}{\mathcal{R}}%

\BibLatexMode{%
   \bibliography{rect_empty}
}

\begin{document}

\title{Orthogonal Emptiness Queries for Random Points}

\author{%
   Jonathan E. Dullerud%
   \JonathanThanks{}%
   \and%
   Sariel Har-Peled%
   \SarielThanks{%
      Work on this paper was partially supported by NSF AF award
      CCF-2317241.  }%
}

\maketitle

\begin{abstract}
    We present a data-structure for orthogonal range searching for
    random points in the plane. The new data-structure uses (in
    expectation) $O\bigl(n \log n ( \log \log n)^2 \bigr)$ space, and
    answers emptiness queries in constant time. As a building block,
    we construct a data-structure of expected linear size, that can
    answer predecessor/rank queries, in constant time, for random
    numbers sampled uniformly from $[0,1]$.

    While the basic idea we use is known \cite{d-lnba-86}, we believe
    our results are still interesting.
\end{abstract}

\section{Introduction}

Orthogonal range searching involve preprocessing a set $S$ of $n$
points in $\Re^d$, such that given a query axis-aligned box, it
quickly decides whether there are any points of $S$ in this box (or
similar queries, such as reporting or counting all such points).
There is vast literature on the topic, see the survey by Agarwal
\cite{a-rs-04}.

Here, we are interested (mostly) in the one and two dimensional cases,
for points randomly sampled from the unit square.

\subsection*{Our results.}

\paragraph{Rank (translation).}

Given a set $\S$ of $n$ numbers in $\Re$, a natural problem is to
preprocess it to answer predecessor queries on it.  For a number $x$,
its \emphw{predecessor} in $\S$ is the largest number $y \in S$, such
that $y \leq x$.  In general, any data structure with
$O(n\, \polylog n)$ space, has a predecessor query time
$\Omega( \log \log n)$ \cite{pt-tstop-06}. We are interested here in a
related problem that can be solved using predecessor search --
specifically, \emphw{rank translation} -- specifically, given a query
number $x$, we want to quickly compute the number of elements in $\S$
that are smaller (or equal) to it. That is, the task is to compute the
\emphw{rank} of $x$:
$\rankY{x}{\S} = \cardin{\Set{s \in \S}{s \leq x}}$.

We show that rank translation (and thus predecessor queries) can be
performed in constant time using (expected) linear space, for $n$
numbers chosen randomly (say from $[0,1]$). This is better than the
lower bound of $\Omega( \log \log n)$ of \Patrascu and Thorup
\cite{pt-tstop-06}. This follows by a two level scheme that is
somewhat reminiscent of Fredman \etal \cite{fks-sstow-84} hashing
scheme. See \secref{rank:queries} for details.

\paragraph{Rank space queries in constant time.}

The above rank translation data-structure, enables us to consider the
input of $n$ random points (that are picked uniformly from the unit
square), are in rank space.

\newcommand{\rankedY}[2]{\mathcalb{r}\pth{#1, #2}}%
\newcommand{\rankedX}[1]{{#1}_\mathcal{R}}%

\begin{defn}
    Let $\P$ be a set of $n$ points in $\Re^d$ with all numbers used
    being distinct. For $i=1,\ldots, d$, let
    $\P_{:i} = \Set{p_i}{(p_1, \ldots, p_d) \in \P}$ denote the
    \emphw{$i$\th slice} of $\P$. The a point
    $p = (p_1, \ldots, p_d) \in \Re^d$, its \emphw{rank} is the
    integral point $\rankedY{\P}{p} = (i_1, \ldots, i_d)$, where
    $i_k = \rankY{p_i}{\P_{:i}}$, for $i=1,\ldots,d$. Thus, the point
    set $\P$ in \emphw{rank space} is the set
    $\rankedX{\P} = \Set{\rankedY{\P}{p}}{p \in \P}$.
\end{defn}

Specifically, if $\P \subseteq \Re^2$ is a set of $n$ points in
general position (i.e., no two points share a coordinate), then
$\rankedX{\P}$ can be represented as a permutation
$\pi \IRX{n} \mapsto \IRX{n}$, where $\IRX{n} = \{1,\ldots, n\}$.

In \secref{rank:space}, we show that a set of $n$ points in rank
space (in 2d), and a prespecified $\eps \in (0,1)$, can be
preprocessed in $O(n^{1+\eps})$ space/time for emptiness queries,
where the emptiness query takes $O(1/ \eps)$ time, using a recursive
high fan-out tree construction on the $x$-axis, followed by a
secondary data-structure on the $y$-axis.

Given a query, the data
structure recursively constrains the query using a tree built upon the
vertical slabs described before. We also find that emptiness quadrant
queries can be pre-processed in $O(n)$ space/time and decided in
$O(1)$ time. This is done using a data structure which looks at the
"minima" curve of a point set, and then checking if the point in the
minima corresponding to the x-coordinate of the quadrant query is
dominated by the y-coordinate of the quadrant query.

Our data structure is comparable to the current fastest pre-processed
data structure of Chan, Larsen and \Patrascu \cite{clp-orsrr-11},
which has pre-processing space/time $O(n \log \log n)$ and has query
time $O( \log \log n)$. Notice also that the emptiness quadrant case
is faster in both pre-processing and query time.

\paragraph{Random Point Queries}

In the case where points are sampled from a random uniform point
distribution, it is clear that one can simply inherit the data
structure from \cite{clp-orsrr-11}, but in fact we show a reduction
that uses the data structure for quadrant queries (Lemma 4.4), as well
as constant $d$-dimensional rank queries (Lemma 3.3). First in Section
5.1, we employ a grid-bucketing method that includes multiple tilings
of the unit square , where each tile is of $\Theta(\log n/n)$
area. That is, each tile "set" is assigned a tile "type" which is
translated to make up the entire unit square, each of differing first
and second dimension (i.e., a tile of type
$[0,\log n/n] \times [0, 1]$ versus a tile of type
$[0,4 \log n/n] \times [0, 1/4]$ ). By eliminating the need to process
"large" query rectangles by the $\eps - $net theorem
\cite{hw-ensrq-87}, we design a data structure that can solve queries
of "small" rectangles which can only intersect a constant number of
rectangles in some grid (the way this particular grid set is chosen is
explained in Section 5.1). We find that one can build this primary
data structure in $O(n \log^{1+\eps} n)$ expected time and space given
a fixed constant $\eps$, which can solve emptiness queries in $O(1)$
time with high probability.

In Section 5.2, we establish a space reduction on individual
rectangles using a 3-level range search tree. This data structure
operates on an individual rectangle such that given an x and y
interval (from a query rectangle), the tree can constrain the query to
four quadrants. By using the quadrant query data structure described
in Section 4.1, one can decide if any of the four quadrants are
empty. This data structure can be pre-processed in $O( n\log^2 n)$
time and space. and decide queries in $O(1)$ time. Then applying this
structure to each bucket in the data structure described in Section 4,
we find that one can pre-process $n$ random points in expected
$O\bigl(n \log n (\log \log n)^2 \bigr)$ expected pre-processing
time/space and expected $O(1)$ query time.

A summary of our results is presented in \figref{results}.

\begin{figure}[h]
    \centering
    \begin{tabular}{|c|c|c|}
      \multicolumn{3}{c}{2d rank space}\\
      \hline
      Space
      &
        Query time
      &
        Ref\\
      \hline
      $O(n \log \log n)\Bigr.$
      &
        $O( \log \log n)$
      &
        \cite{clp-orsrr-11}
      \\
      $O(n^{1+\eps})\Bigr.$
      &
        $O( 1/\eps)$
      &
        \lemref{n:1:eps}
      \\
      \hline
      \multicolumn{3}{c}{}\\
      \multicolumn{3}{c}{2d random points}\\
      \hline%
      $O(n \log^{1+\eps} n)\Bigr.$ & $O(1/\eps)$ & \thmref{main:1}\\
      $O\bigl(n \log n (\log \log n)^2 \bigr)\Bigr.$ & $O(1)$ & \thmref{main:2}\\
      \hline
    \end{tabular}
    \caption{Results on orthogonal range searching in 2d in rank space
       and random points.}%
    \figlab{results}%
\end{figure}

\paragraph{Previous work.}

The basic idea about predecessor search is already known, see Devroye
\cite{d-lnba-86}.  There is also associated work interpolation search
and how to make it dynamic \cite{mt-dis-93} which is relevant.
Nevertheless, we believe our writeup is still interesting. In
particular, it seems likely one can further improve the space used by
our data-structure in the 2d case.

\section{Preliminaries}
\seclab{prelims}

For integers $i<j$, let $\IRY{i}{j} = \{ i, i+1, \ldots, j\}$, and let
$\IRX{i} = \IRY{1}{i} = \{1,\ldots,i \}$.

\begin{defn}
    For a set $\S$ of $n$ distinct real numbers, the \emphi{rank} of a
    number $x \in \Re$, is
    $\rankY{x}{\S} = \cardin{\Set{s \in S}{s \leq x}}$.
\end{defn}

\paragraph{Rectangles.}
For technical reasons it is easier to work with semi-open
rectangles. All the rectangles we work with are axis-aligned.
\begin{defn}
    A \emphi{rectangle} is the semi-open set of the form
    $(x,x'] \times (y,y']$, for any real numbers $x,x',y,y'$.
\end{defn}

Let $\P$ be a set of points in the plane. Consider the natural
translation of this point set into \emphi{rank space}:
$\rankX{\P} = \Set{ \rankY{p}{\P}}{p \in \P}$. If the points of $\P$
all have distinct coordinates, then $\rankX{\P}$ can be viewed as a
permutation of $\IRX{n}$.

\begin{defn}
    \deflab{rank:translation}%
    For a point set $\P$, and a rectangle
    $\rect = (x,x'] \times (y,y']$, its \emphi{rank translation} is
    the version of $\rect$ in the rank space:
    $\rankY{\rect}{\P} = (\alpha, \alpha'] \times (\beta,\beta']$,
    where $(\alpha, \beta) = \rankY{\bigl.(x,y)}{ \P}$ and
    $(\alpha', \beta') = \rankY{\bigl.(x',y')}{ \P}$.
\end{defn}

\begin{problem}
    The input is a set $\P$ of $n$ points randomly, uniformly and
    independently sampled from the unit square. Our problem is to
    construct a data structure which can decide whether an
    axis-aligned query rectangle contains at least one point of $\P$.
\end{problem}

\section{Rank queries for random numbers and points}
\seclab{rank:queries}

\paragraph{Rank queries.}
Let $\S$ be a set of $n$ numbers uniformly and independently sampled
from the interval $[0,1]$. The task is to construct a data structure
for rank queries on $\S$.  constant time.

\paragraph{Construction.}

Consider throwing $m$ points uniformly into the interval $[0,1]$, and
let $\alpha$ be an small integer. We partition the interval $[0,1]$
into $m^\alpha$ equal length subintervals (i.e., bins). We consider a
bad event, when (say) $5$ (or more) points fall into the same bin.

\begin{lemma}
    \lemlab{top:level}%
    Let $\S$ be a set of $m$ points chosen uniformly and independently
    from an internal $I$. Then, one can construct a data-structure of
    expected size $O(m^2)$, in $O(m^2)$ expected time, such that one
    can report the rank of a query $x \in I$ in $O(1)$ time.
\end{lemma}

\begin{proof}
    For an $\alpha > 0$ an integer, the basic scheme is to partition
    $I$ into $m^\alpha$ (equal length) subintervals (i.e., bins), such
    that no bin contains more than $4$ points of $\S$. Clearly, for
    $\alpha$ sufficiently large this happens with good probability, as
    we show below. Assume that this in indeed the case, and observe
    that this readily leads to a data-structure that can resolve rank
    queries in constant time. Indeed, sort the points of $\S$ into
    their $m^\alpha$ bins (using radix sort in $O(m)$ time, say), and
    now compute for each bin how many points fell into it (again, this
    is at most $4$). Now compute the prefix sums for the sizes of the
    bins. Given a query $x \in I$, we compute in constant time (using
    the floor function, say), the bin containing $x$, use the prefix
    sum computed for the previous bin, together with the rank of $x$
    among the (four) numbers stored in this bin, to compute the rank
    of $x$. This takes $O(1)$ time, and requires $O(m^\alpha)$ space.

    We next the expected space used by this scheme.  The probability
    that two specific balls (when distributed to their respective
    bins) collide is $p_\alpha =1/m^\alpha$. Let $X_\alpha$ be the
    number of $5$ tuple of balls that collide. We have that
    \begin{equation*}
        \beta_\alpha
        =%
        \Ex{X_\alpha}%
        =%
        \binom{m}{5} p_\alpha^4 \leq%
        m^5 \frac{1}{m^{4\alpha}}.
    \end{equation*}
    The probability that any bin contains more than $4$ balls is thus,
    by Markov's inequality, at most
    $\Prob{X_\alpha \geq 1 } \leq \Ex{X_\alpha} \leq \beta_\alpha$.
    Let $Z$ be the minimum value of $\alpha$, such that $\S$ contains
    no bin with more than $4$ points. We have that
    $\Prob{Z \geq i} \leq \Prob{X_{i-1} \geq 1 } \leq \min( 1,
    m^{5-4(i-1)})$.  The expected space used by the above
    data-structure is thus bounded by $O(\Delta)$, where
    \begin{equation*}
        \Delta
        \leq%
        \sum_{i=1}^{\infty} \Prob{Z\geq i} m^i%
        \leq%
        m + m^2 +
        \sum_{i=3}^\infty m^i \cdot m^{5-4(i-1)}
        \leq%
        m + m^2 +
        \sum_{i=0}^\infty \frac{1}{m^{4i}}
        \leq
        2m^2.
    \end{equation*}
\end{proof}

\begin{lemma}
    \lemlab{1:rank}%
    Let $\S$ be a set of $n$ points chosen uniformly and independently
    from an internal $I \subseteq \Re$. Then, one can construct a
    data-structure of expected size $O(n)$, in $O(n)$ expected time,
    such that one can report the rank of a query $x \in I$ in $O(1)$
    time.
\end{lemma}
\begin{proof}
    The idea is to apply a two-level scheme -- in the first level, we
    use $n$ bins, as we have $n$ numbers, and then in the second
    level, for each bin containing five or more points, we use the
    data-structure of \lemref{top:level}.

    Sorting the points into the bins can be done in $O(n)$ time using
    radix-sort together with the floor function.
    Let $n_i$ be the number of points falling into the $i$\th bin in
    the top level. We have that the expected construction time (and
    space) is
    \begin{math}
        \Delta = O( n + \sum_{i=1}^{n} n_i^2 ).
    \end{math}
    To bound this quantity (in expectation), let $Y$ be the number of
    points falling into the first bin. We have that
    $Y \sim \mathrm{Bin}(n, 1/n)$. Thus,
    $\Prob{Y=i} \leq \binom{n}{i} \frac{1}{n^i} \leq
    \frac{1}{i!}$. Thus, we have
    \begin{math}
        \Ex{Y^2}
        =%
        \sum_{i=1}^\infty \frac{1}{i!} i^2
        =%
        O(1).
    \end{math}
    This readily implies that

    more than $5$, balls we use the
    above scheme. This implies that, for all $i$, we have
    $\Ex{O(n_i^2)} = O(1)$, and thus, the total space used by the
    data-structure is $O(\Delta ) = O(n)$.

    As for answering a query, we pre-compute prefix sums on the bins
    in the top level. Given a query it is thus sufficient to resolve
    its rank inside its bin, which can be done in $O(1)$ time, by
    \lemref{top:level}.
\end{proof}

\paragraph*{Rank queries in higher dimensions.}

An immediate consequence of the above, is that if $\P$ is a set of $n$
points picked randomly (uniformly and independently) from $[0,1]^d$,
then we can answer rank queries on $\P$ in $O(1)$ time.

As a reminder,
for a set of points $\P$ in $\Re^d$, and a query point $q \in \Re^d$,
a \emphw{rank query} returns an integral point $(i_1,\ldots, i_d)$,
where $i_k$ is the rank of $q_k$ among the $n$ numbers forming the
$k$\th coordinate of the points of $P$. We thus immediately get the
following:

\begin{lemma}
    \lemlab{d:rank}%
    Let $\P$ be a set of $n$ points chosen uniformly and independently
    from $[0,1]^d$. Then, one can construct a data-structure of
    expected size $O(n)$, in $O(n)$ expected time, such that one can
    answer rank queries on $\P$ in $O(1)$ time.
\end{lemma}

\section{Emptiness queries in rank space}
\seclab{rank:space}

Here, we describe how to do emptiness queries in rank space in
constant time for points in the plane. The following is by now
standard (i.e., a tree with high fan-out), and the full description is
included here for the sake of completeness.

\subsection{Rectangle queries}

Let $\P$ be a set of $n$ points in rank space. That is,
$\P = \{p_1, \ldots, p_n\}$, with $p_i = (i,y_i)$, where
$i, y_i \in \IRX{n}$, and no two points of $\P$ have the same
$y$-coordinate. We are interested in answering emptiness rectangle
queries, where the rectangles are of the form
$(i_1,i_2] \times (j_1, j_2)$, where $i_1,i_2,j_1,j_2 \in \IRY{0}{n}$.

\paragraph{The easy solution: \TPDF{$O(n^4)$}{n^4} space.}

By recomputing the query result over all possible queries, one readily
gets the following\footnote{Getting the $O(n^4)$ preprocessing time
   requires more work, but since we present a better result shortly,
   we omit the interesting but irrelevant details.}.
\begin{lemma}
    Let $\P$ be a set of $n$ points in rank space with unique
    coordinates. One can preprocess $\P$ in $O(n^4)$ time and space,
    so that rectangle emptiness queries are answered in constant time.
\end{lemma}

\paragraph{Reducing the space.}

For integers $\alpha \leq \beta$, let
$\Q = \P \sqcap [\alpha,\beta] = \Set{ (i,j) \in \P}{ \alpha \leq i
   \leq \beta }$. A rectangle $\qrect = (i_1,i_2] \times (j_1, j_2]$
\emphi{crosses} $\Q$ if $i_1 < \alpha$ and $i_2 \geq \beta$.
\begin{lemma}
    \lemlab{crossing}%
    Given a set $\P$ of $n$ points in rank space, and an interval
    $[\alpha,\beta]$, one can preprocess
    $\Q = \P \sqcap [\alpha,\beta]$, using $O(n)$ space and time, such
    that given a crossing rectangle $\rect$, one can decide if $\rect
    \cap \Q$ is empty in $O(1)$ time.
\end{lemma}
\begin{proof}
    Compute the $y$-rank of $j_1$ and $j_2$ in $\Q$ -- let $j_1'$ and
    $j_2'$ denote these ranks. Clearly $Q \cap \qrect = \emptyset$
    $\iff$ $j_1' \geq j_2'$.  Computing the ranks is straightforward
    -- set an array $A[1\cdots n]$, where $A[y_j]=1$ $\iff$
    $(j,y_j) \in \Q$, and $0$ otherwise. The rank of an integer
    $j_2 \in \IRX{n}$ is no more than the number of $1$s in
    $A[1\cdots j_2]$. To this end, compute the prefix sums of $A$,
    with $B[i]= \sum_{j=1}^i A[j]$. The rank of $j_2$ is
    $B[j_2]$. Thus, all one need to do is to store the array $B$, and
    the rank resolution queries can be performed in constant time.
\end{proof}

\begin{lemma}
    \lemlab{n:1:eps}%
    Let $\P$ be a set of $n$ points in rank space in the plane. For
    any constant $\eps \in (0,1)$, one can preprocess $\P$ in
    $O(n^{1+\eps})$ time/space, such that one can answer rectangle
    emptiness queries in $O(1/\eps)$ time.
\end{lemma}
\begin{proof}
    Assume $\P = \{p_1, \ldots, p_n\}$ be the $n$ input points sorted
    by their $x$-coordinate. Let $\Delta = \ceil{n^{1 - \eps/2}}$, and
    assume for simplicity that $\Delta$ divides $n$, and let
    $m = n/\Delta \leq n^{\eps/2}$.  We partition $P$ into $m$
    vertical slabs each containing $\Delta$ points.  For
    $\alpha < \beta$ integers, let
    $\P(\alpha,\beta) = \{ p_{(\alpha-1)\Delta +1} , \ldots,
    p_{\beta\Delta} \}$.  For all $i ,j \in \IRX{m}$, with $i \leq j$,
    build the crossing rectangles emptiness query data-structure of
    \lemref{crossing}. This requires $O(m^2 n) = O(n^{1+\eps})$
    space. For each $i$, let $P_i = P(i,i)$, and construct the same
    data-structure recursively for each $P_i$ -- this construction
    continues recursively with fan-out degree $m$ -- resulting in a
    tree of depth $h = O(1/\eps)$. The preprocessing time and space is
    $O( n^{1+\eps})$. In the leaf of this tree, corresponds to a
    singleton point of $\P$.

    Given a query rectangle $\qrect = (i_1,i_2] \times (j_1, j_2]$,
    one computes the minimum $\alpha$ and minimum $\beta$, such that
    $\qrect$ crosses $\P(\alpha,\beta)$. One perform emptiness query
    for $\qrect$ in the data-structure computed for
    $P(\alpha,\beta)$. Next, one continues the query in the two
    children constructed for $P_{\alpha-1}$ and $P_{\beta+1}$. If any
    of the $O(1/\eps)$ emptiness queries performed by this process in
    the data-structures constructed of \lemref{crossing} returns the
    rectangle is not empty, the process stops and return this
    answer. Otherwise, it returns that the query rectangle is
    empty. Clearly, the query time is $O(1/\eps)$.

    Note, that as the query descends, one need to do rank translation
    of the query rectangle into the point set in the child. This can
    be readily done by constructing the appropriate data-structure as
    described in the proof of \lemref{crossing}.
\end{proof}

\newcommand{\minimaX}[1]{\mathrm{minima}\pth{#1}}%

\subsection{Quadrant queries}

Let $\P$ be a of $n$ points in rank space, and consider an empty
quadrant query -- that is, a query rectangle of the form
$(0,i_2] \times (0,j_2])$ (the other three possible kinds of quadrants
are handled in a similar fashion). A point $ p = (x,y)$
\emphi{strictly dominates} a point $p' = (x',y')$, denoted by
$p \succ p'$, if $x > x'$ and $y > y'$. If $x \geq x'$ and $y \geq y'$
then $p$ \emphi{dominates} $p'$. The \emphi{minima} of $\P$, denoted
by $\minimaX{\P}$, is the curve formed by all the points $(i,f(i))$,
for $i=1, \ldots, n$, that dominates some point of $\P$, but do not
strictly dominates any point of $\P$, and $f(i)$ is minimized among
all $y$ values with this property.

Since sorting of the points of $\P$ by the $x$-axis can be done in
linear time using radix sort, the minima curve of $\P$ can be computed
in linear time (its essentially a scan with a heap similar in spirit
to computing the convex-hull of sorted points). We store the
$y$-coordinates of the $n$ points forming the minima in an array
$F[1\ldots n]$,

Answering an emptiness quadrant query
$\qrect = (0,i_2] \times (0,j_2])$ is done by checking if
$F[i_2] \leq j_2$. If it is, then $\qrect \cap \P$ is not
empty. Otherwise it is. Clearly, this data-structure can be
constructed in linear time, using linear space, and can answer
quadrant emptiness queries in constant time, implying the following.

\begin{lemma}
    \lemlab{quadrant:r:space}%
    Given a set $\P$ of $n$ points in rank space in two dimensions,
    one can preprocess them in $O(n)$ time and space, such that given
    a query quadrant $\qrect = (0,i_2] \times (0,j_2]$, one can decide
    if $\qrect \cap \P$ is empty in $O(1)$ time.
\end{lemma}

\section{Emptiness queries for random points}

For a set $\P$ of $n$ points in \emph{rank space} in the plane, Chan
\etal \cite{clp-orsrr-11} presented a data-structure with
$O(n \log \log n)$ space, for answering orthogonal 2d emptiness range
queries in $O( \log \log n)$ time (note that the query is also in rank
space). Using \lemref{d:rank} we can now map random points to rank
space in constant time per query. This readily implies that we get the
same bound for random points.  However, one can do significantly
better.

\subsection{Bucketing and a direct reduction to rank space}

The $\eps$-net theorem
\cite{hw-ensrq-87}  readily implies the following.
\begin{lemma}
    \lemlab{big:rectangles}%
    For any $c > 16$, we have for a random sample $\P$ of $n$ points
    from $[0,1]^2$, any axis aligned rectangle $R \subseteq [0,1]^2$
    of area $\geq c(\log n)/n$, has that $R \cap P \neq
    \emptyset$. This property holds for all such rectangles $R$ with
    probability $\geq 1- 1/n^{\Omega(c)}$.
\end{lemma}

This implies that for random points, for emptiness queries, one needs
to handle only the ``small'' rectangles.  Let
$\TwoCeil{x} = 2^{\ceil{\log_2 x}}$ be the smallest power of $2$ at
least as large as $x$.  Let $N = \TwoCeil{n}$,
$\lN = \ceil{\log_2 n} = \log_2 N$, and let $\cGrid > 16$ be a
sufficiently large constant. Let
$w = \TwoCeil{4 \cGrid \log_2 N} = \Theta(\log n)$. For
$i=0, \ldots, \log_2 \frac{N}{w} = \lN - \log_2 w$, consider the
rectangle
\begin{equation*}
    \rect_i = [0,2^i w/N] \times [0, 1/2^i],
\end{equation*}
that is of area $\geq \cGrid (\log_2 n)/n$. Let $G_i$ be the
$N/(w 2^i) \times 2^i$ grid tiling of $[0,1]^2$ by translated copies
of $r_i$. Let $R_i$ be this set of rectangles forming $G_i$. Observe
that $\cardin{R_i} = O(n / \log n)$. By \lemref{big:rectangles} all
the rectangles of $R_i$, for all $i$, are not empty.

\paragraph{Data-structure.}
Fix a constant $\eps \in (0,1)$.  Let $\RectSet = \cup_i R_i$. For
every rectangle $\rect \in \RectSet$, we build the data-structure of
\lemref{d:rank} to answer rank queries for the point set
$\P_\rect = \P \cap \rect$, and the data-structure of \lemref{n:1:eps}
on the rank point set.

\paragraph{Answering a query.}

Given a query rectangle $\qrect \subseteq [0,1]^2$, if its area
exceeds $w /N$, then the data-structure returns that $\qrect$ is not
empty.  Observe that $\qrect$ can not contain fully any of the
rectangles of $\RectSet$.

Otherwise, $\ell_X$ and $\ell_y$ be the lengths of the projections of
$\qrect$ on the $x$ and $y$ coordinates, respectively. Let $j$ be the
minimum integer such that $\ell_x \leq 2^j w/N$. As
$\areaX{\qrect} \leq \areaX{\rect_j}$, for any rectangle
$r_j \in R_j$, it follows that $\qrect$ can intersect at most $6$ grid
cells of $R_j$. For each of these grid cells, we now answer the
emptiness query for $\qrect$ inside the cell using the above
precomputed data-structures. Indeed, for a cell $\rect_j$ of this grid
that intersects $\qrect$, we compute the rank rectangle
$\rankY{\qrect}{\P \cap \rect_j}$, and then perform the query on the
data-structure of \lemref{n:1:eps} constructed for the point set
$\P \cap \rect_j $ in rank space.

\paragraph{Space/preprocessing.}

For a rectangle $\rect \in \RectSet$, for that
$n_\rect = \cardin{\P \cap \rect}$, we have
$\Ex{n_\rect}= \Theta( \log n)$. Chernoff's inequality immediately
implies that $n_\rect =\Theta( \log n)$, with high probability, for
all the rectangles in $\RectSet$. For each rectangle in $\RectSet$,
the preprocessing time and space is $O(n_\rect^{1+\eps})$. Since
$\cardin{\RectSet} = O(n)$, we conclude that the overall preprocessing
takes $O( n \log^{1+\eps} n)$ time and space.

We summarize the result.
\begin{theorem}
    \thmlab{main:1}%
    For a set $\P$ of $n$ random points picked uniformly and
    independently from $[0,1]^2$, and a fixed constant
    $\eps \in (0,1)$, one can build a data-structure in
    $O(n \log^{1+\eps} n)$ expected time and space, such that given a
    query axis-aligned rectangle $\qrect$, one can decide if
    $\qrect \cap \P$ is empty in $O(1)$ time. This data-structure
    result is correct with high probability (for all possible
    queries).
\end{theorem}

\subsection{Reducing the space further}

Intuitively, the bottleneck in the space used in the above
data-structures in rank space is \defrefY{rank
   translation}{rank:translation}. Specifically, computing the ranks
of the vertices of a rectangle $\rect$, given in the rank coordinates
of $\P$, into the rank coordinates of $\Q \subseteq \P$, requires
$\Omega(|\P|)$ space (see \lemref{crossing}). However, for random
points, one can do this translation (under the right settings) using
$O(|\Q|)$ space (see \lemref{d:rank}).

\newcommand{\Tree}{\mathcal{T}}%

To this end, we build a ``traditional'' range search tree on the
(randomly chosen) points of $\P$. The top tree $\Tree$ is built on the
points of $\P$ sorted by their $x$-coordinate. The secondary tree is
built for each node $v$ of the top tree, sorted by the $y$-coordinates
of the points stored in the subtree of $v$. As the third level
data-structure, for each node, we build the data-structure of
\lemref{quadrant:r:space} that answers quadrant emptiness queries in
the rank space of the points in this subproblem (specifically, four
such data-structures for each of the possible quadrants), and also the
data-structure of \lemref{d:rank} to do rank translation. We show
below that one can construct such data-structure that uses linear
space in the number of points in the point set.

\paragraph{Answering a query.}
Given a query rectangle $\qrect = I_x \times I_y$, we first do rank
translation for its $x$-range $I_x$, built for $\P$. This takes $O(1)$
time, and results in a semi-open interval $(i_1, i_2]$ where
$i_1, i_2 \in \IRY{0}{n}$. We can consider $\Tree$ to be a complete
balanced binary tree over $\IRX{n}$ (e.g., consider $n$ to be a power
of $2$). Thus, it is straightforward to preprocess it for \LCA queries
that are answered in constant time (and using linear space). Let $u$
be this common ancestor, and let $v,w$ be its two children. The query
now continues in the secondary tree built for $\P_v$ and $\P_w$ (i.e.,
in the children). In this secondary tree the query now is a three
sided rectangle. We now repeat the same process in two secondary
trees. Here, we describe it for $\Tree_v$. We do the rank translation
of $I_y$ for the $y$-coordinates of the points of $\P_v$ (importantly,
the $y$ coordinates of the points of $\P_v$ are independent of the $x$
coordinates of these points, so we can still use the data-structure of
\lemref{d:rank} to do this in constant time). Again, using \LCA query,
in constant time, wen can locate the node in $\Tree_v$ that splits the
$y$-interval of $\qrect$. As above, we continue the query in the two
children. This overall, results in the query process continuing in
four third level data-structure, in each one of them, the query is now
a quadrant.

A node $u$ in the secondary data-structure corresponds to a
rectangular region $\rect_u$. Beyond knowing the number of points of
$\P$ inside this rectangle, the points $\P \cap\rect_u$ are uniformly
distributed in $\rect_u$. As such, for $\qrect$ we can do the rank
translation $\qrect' = \rankY{\qrect}{ \P \cap \rect_u}$ in constant
time (again, using \lemref{d:rank}). Finally, we perform the quadrant
emptiness query for $\qrect'$ on the rank space point-set
$\P \cap \rect_u$, which takes constant time.

\paragraph{Space and preprocessing.}
The third level data-structure requires linear space in the number of
points it is build for. As such, the second level tree built for a
node $u \in \Tree$, requires $O( \cardin{\P_u} \log
\cardin{\P_u})$. Thus, as in the standard analysis of range-trees, we
have that the overall space used is $O(n \log^2 n)$.

\begin{lemma}
    Given a set $\P$ of $n$ random points picked uniformly and
    independently from $[0,1]^2$, one can preprocess them in expected
    $O( n\log^2 n)$ time and space, such that given a query rectangle
    $\qrect$, one can decide if $\qrect \cap \P$ is empty in $O(1)$
    time.
\end{lemma}

But wait! We can now use the above data-structure in each of the
buckets of the grids of \thmref{main:1}. We thus get the following
result.
\begin{theorem}
    \thmlab{main:2}%
    For a set $\P$ of $n$ random points picked uniformly and
    independently from $[0,1]^2$, one can build a data-structure in
    $O\bigl(n \log n (\log \log n)^2\bigr)$ expected time and space,
    such that given a query axis-aligned rectangle $\qrect$, one can
    decide if $\qrect \cap \P$ is empty in $O(1)$ time. This
    data-structure result is correct with high probability (for all
    possible queries).
\end{theorem}

\paragraph*{Acknowledgement.}

The authors thank Timothy Chan for pointing out the work by Devroye
\cite{d-lnba-86}.

\BibTexMode{%
   \bibliographystyle{alpha}
   \bibliography{rect_empty}
}%
\BibLatexMode{\printbibliography}

\end{document}